\documentclass[aps,prb,groupedaddress]{revtex4-1}
\usepackage[utf8]{inputenc}
\usepackage[english]{babel}
\usepackage{amsmath,amssymb,amsfonts,amsthm}
\usepackage{fullpage}
\usepackage{graphicx,color}

\newtheorem{theorem}{Theorem}

\newtheorem{remark}{Remark}
\newtheorem{lemma}{Lemma}

\newenvironment{myproof}[1][\proofname]{%
  \proof[#1]%
}{\endproof}

\def\N{{\mathbb N}}
\def\Z{{\mathbb Z}}

\def\R{{\mathbb R}}
\def\C{{\mathbb C}}

\def\1{{\mathds{1}}}

\begin{document}

\title{Discretization error cancellation in electronic structure calculation: \\ a quantitative study}

\author{Eric Canc\`es\footnote{\texttt{cances@cermics.enpc.fr}}}
\affiliation{CERMICS, Ecole des Ponts and INRIA Paris, 6 \& 8 Avenue Blaise Pascal, 77455 Marne-la-Vall\'ee, France}

\author{Genevi\`eve Dusson\footnote{\texttt{dusson@ljll.math.upmc.fr}}}
\affiliation{Sorbonne Universit\'es, UPMC Univ. Paris 06 and CNRS, UMR 7598, Laboratoire Jacques-Louis Lions, F-75005, Paris, France, and Sorbonne Universit\'es, UPMC Univ. Paris 06,  Institut du Calcul et de la Simulation, F-75005, Paris, France}

\date{\today}

\begin{abstract} 
It is often claimed that error cancellation plays an essential role in quantum chemistry and first-principle simulation for condensed matter physics and materials science. Indeed, while the energy of a large, or even medium-size, molecular system cannot be estimated numerically within chemical accuracy (typically 1 kcal/mol or 1 mHa), it is considered that the energy difference between two configurations of the same system can be computed in practice within the desired accuracy.

The purpose of this paper is to provide a quantitative study of discretization error cancellation. The latter is the error component due to the fact that the model used in the calculation (e.g. Kohn-Sham LDA) must be discretized in a finite basis set to be solved by a computer. We first report comprehensive numerical simulations performed with Abinit~\cite{Gonze2009,Gonze2002} on two simple chemical systems, the hydrogen molecule on the one hand, and a system consisting of two oxygen atoms and four hydrogen atoms on the other hand. We observe that errors on energy differences are indeed significantly smaller than errors on energies, but that these two quantities asymptotically converge at the same rate when the energy cut-off goes to infinity. We then analyze a simple one-dimensional periodic Schr\"odinger  equation with Dirac potentials, for which analytic solutions are available. This allows us to explain the discretization error cancellation phenomenon on this test case with quantitative mathematical arguments.
\end{abstract}

\pacs{}

\maketitle

\section{Introduction}

Error control is a central issue in molecular simulation. The error between the computed value of a given physical observable (e.g. the dissociation energy of a molecule) and the exact one, has several origins. First, there is always a discrepancy between the physical reality and the reference model, here the $N$-body Schr\"odinger equation, possibly supplemented with Breit terms to account for relativistic effects. However, at least for the atoms of the first three rows of the periodic table, this reference model is in excellent agreement with experimental data, and can be considered as exact in most situations of interest. The overall error is therefore the sum of the following components:
\begin{enumerate}
\item the {\em model error}, that is the difference between the value of the observable for the reference model, which is too complicated to solve in most cases, and the value obtained with the chosen approximate model (e.g. the Kohn-Sham LDA model), assuming that the latter can be solved exactly;
\item the {\em discretization error}, that is the difference between the value of the observable for the approximate model and the value obtained with the chosen discretization of the approximate model. Indeed, the approximate model is typically an infinite dimensional minimization problem, or a system of partial differential equations, which must be discretized to be solvable by a computer, using e.g. a Gaussian atomic basis set, or a planewave basis;
\item the {\em algorithmic error}, which is the difference between the value of the observable obtained with the exact solution of the discretized approximate model, and the value computed with the chosen algorithm. The discretized approximate models are indeed never solved exactly; they are solved numerically by iterative algorithms (e.g. SCF algorithms, Newton methods), which, in the best case scenario, only converge in the limit of an infinite number of iterations. In practice, stopping criteria are used to exit the iteration loop when the error at iteration $k$, measured in terms of differences between two consecutive iterates or, better, by some norm of some residual, is below a prescribed threshold. If the stopping criterion is very tight, the algorithmic error can become very small, ... or not! For instance, if the discretized approximate model is a non convex optimization problem, there is no guarantee that the numerical algorithm will converge to a global minimum. It may converge to a local, non-global minimum, leading to a non-zero algorithmic error even in the limit of an infinitely tight stopping criterion;
\item the {\em implementation error}, which may, obviously, be due to bugs, but does not vanish in the absence of bugs, because of round-off errors: in molecular simulation packages, most operations are implemented in double precision, and the resulting round-off errors can accumulate, especially for very large systems;
\item the {\em computing error}, due to random hardware failures (miswritten or misread bits). This component of the error is usually negligible in today's standard computations, but is expected to become critical in future exascale architectures~\cite{Li11}.
\end{enumerate}
Quantifying these different sources of errors is an interesting purpose for two reasons. First, guaranteed estimates on these five components of the error would allow one to supplement the computed value of the observable returned by the numerical simulation with guaranteed error bars (certification of the result). Second, they would allow one to choose the parameters of the simulation (approximate model, discretization parameters, algorithm and stopping criteria, data structures, etc.) in an optimal way in order to minimize the computational effort required to reach the target accuracy.

\medskip

The construction of guaranteed error estimators for electronic structure calculation is a very challenging task. Some progress has however been made in the last few years, regarding notably the discretization and algorithmic errors for Kohn-Sham LDA calculations.  {\it A priori} discretization error estimates have been constructed in~\cite{M2AN12} for planewave basis sets, and then in~\cite{Zhou2013} for more general variational discretization methods. {\it A posteriori} error estimators of the discretization error have been proposed in \cite{JCP,Chen14,Lin15}. A combined study of both the discretization and algorithmic errors was published in~\cite{Cances2014} (see also \cite{Dusson2016}). We also refer to~\cite{Maday2003,chen2015,Chen15,Kutzelnigg,Lin2016,Hanrath2008,Kutzelnigg1991,Pernot2015,Pieniazek2008,Rohwedder2013} and references therein for other works on error analysis for electronic structure calculation.

In all the previous works on this topic we are aware of, the purpose was to estimate, {\em for a given nuclear configuration $R$ of the system}, the difference between the ground state energy $E_R$ (or another observable) obtained with the continuous approximate model under consideration (e.g. Kohn-Sham LDA) and its discretized counterpart denoted by $E_{R,N}$, where $N$ is the discretization parameter. The latter is typically the number of basis functions in the basis set for local combination of atomic orbitals (LCAO) methods~\cite{Helgaker}, the inverse fineness of the grid or the mesh for finite difference (FD) and finite element (FE) methods~\cite{Gygi95,Saad10,Pask05,Gavini13}, the cut-off parameter in energy or momentum space for planewave (PW) discretization methods~\cite{Gonze2009,QuantumEspresso,VASP}, or the inverse grid spacing and the coarse and fine region multipliers for wavelet (WL) methods~\cite{BigDFT}.  In variational approximation methods (LCAO, FE, PW, and WL), the discretization error $E_{R,N}-E_R$ is always nonnegative by construction. In systematically improvable methods (FD, FE, PW, and WL), this quantity goes to zero when $N$ goes to infinity with a well-understood rate of convergence depending on the smoothness of the pseudopotential~(see \cite{M2AN12} for the PW case). However, in most applications, the discretization parameters are not tight enough for the discretization error to be lower than the target accuracy, which is typically of the order of 1 kcal/mol or 1 mHa (recall that 1 mHa $\simeq$ 0.6275 kcal/mol $\simeq$ 27.2 meV, which corresponds to an equivalent temperature of about 316 K). It is often advocated that this is not an issue since the real quantity of interest is not the value of the energy $E_R$ for a particular nuclear configuration $R$, but the energy difference $E_{R_1}-E_{R_2}$ between two different configurations $R_1$ and $R_2$. It is indeed expected that 
$$
|(E_{R_1,N}-E_{R_2,N}) - (E_{R_1}-E_{R_2}) | \ll  |E_{R_1,N}-E_{R_1}|+|E_{R_2,N}-E_{R_2}|,
$$
that is, the numerical error on the energy difference between the two configurations is much smaller than the sum of the discretization errors on the energies of each configuration. This expected phenomenon goes by the name of (discretization) error cancellation in the Physics and Chemistry literatures.

\medskip

Obviously, for variational discretization methods, $E_{R_j,N}-E_{R_j} \ge 0$ so that both discretization errors have the same sign, leading to 
\begin{align*}
|(E_{R_1,N}-E_{R_2,N}) - (E_{R_1}-E_{R_2}) | &=  \left|(E_{R_1,N}-E_{R_1})-(E_{R_2,N}-E_{R_2})\right| \\ &\le \max \left( E_{R_1,N}-E_{R_1},E_{R_2,N}-E_{R_2} \right),
\end{align*}
but this does not explain the magnitude of the error cancellation phenomenon. The commonly admitted {\em qualitative} argument usually raised to explain this phenomenon is that the errors $E_{R_1,N}-E_{R_1}$ and $E_{R_2,N}-E_{R_2}$ are of the same nature and almost annihilate one another. 

\medskip

The purpose of this article is to provide a {\em quantitive} analysis of discretization error cancellation for PW discretization methods. First, we report in Section~\ref{sec:numerics} two systematic numerical studies on, respectively, the hydrogen molecule and a simple system consisting of six atoms. For these systems, we are able to perform very accurate calculations with high PW cut-offs, which provide excellent approximations of the ground state energy $E_R$. We then compute, for two different configurations $R_1$ and $R_2$, the error cancellation factor 
$$
0 \le Q_N := \frac{|(E_{R_1,N}-E_{R_2,N}) - (E_{R_1}-E_{R_2}) |}{|E_{R_1,N}-E_{R_1}|+|E_{R_2,N}-E_{R_2}|} \le 1.
$$
We observe that this ratio is indeed small (typically between $10^{-3}$ and $10^{-1}$ depending on the system and on the configurations $R_1$ and $R_2$), and that it does not vary much with $N$. 
In Section \ref{sec:toy_problem}, we introduce a toy model consisting of seeking the ground state of a one-dimensional linear periodic Schr\"odinger equation with Dirac potentials:
$$
 \left( -\frac{d^2}{dx^2} -  \sum_{m \in \Z} z_1 \delta_m - \sum_{m \in \Z} z_2 \delta_{m+R} \right) u_R = E_R u_R, \qquad  \int_0^1 u_R^2(x) dx = 1,
$$
for which we can prove that the error cancellation factor $Q_N$ converges to a fixed number $0 < Q_\infty < 1$ when $N$ goes to infinity. Interestingly, it is possible to obtain a simple explicit expression of $Q_\infty$, which only depends on $z_1$, $z_2$ and on $u_{R_1}(0)^2$, $u_{R_2}(0)^2$, $u_{R_1}(R_1)^2$, $u_{R_1}(R_2)^2$, i.e. on the values of the densities $\rho_{R_1}=u_{R_1}^2$ and $\rho_{R_1}=u_{R_2}^2$ at the singularities of the potential.

\section{Discretization error cancellation in planewave calculations}
\label{sec:numerics}

We present here some numerical simulations on two systems: the $H_2$ molecule and a system consisting of two oxygen atoms and four hydrogen atoms. The simulations are done in a cubic supercell of size 10$\times$10$\times$10 bohrs with the Abinit simulation package~\cite{Gonze2009,Gonze2002}. The chosen approximate model is the periodic Kohn-Sham LDA model~\cite{Kohn1965} with the parametrization and the pseudopotential proposed in~\cite{Goedecker1996}. For each configuration $R$, we compute a reference ground state energy $E_R$ taking a high energy cutoff $E_{\rm cut}=400$ Ha. We then compute approximate energies for $N=E_{\rm cut}$ varying from 5 to 105 Ha by steps of 5 Ha. The so-obtained energies are denoted by~$E_{R,N}$.

For two given configurations $R_1$ and $R_2$ of the same system, we  compute $S_N$, the sum of the discretization errors on the energies of the two configurations (note that $E_{R,N}-E_R \ge 0$ since PW is a variational approximation method), and $D_N$, the discretization error on the energy difference:
\[
	S_N = (E_{R_1,N}-E_{R_1})+(E_{R_2,N}-E_{R_2}) \quad \text{and} \quad D_N = \left|(E_{R_1,N}-E_{R_2,N})-(E_{R_1}-E_{R_2})\right|,
\]
as well as the error cancellation factor
$$
Q_N=\frac{D_N}{S_N} = \frac{\left|(E_{R_1,N}-E_{R_2,N})-(E_{R_1}-E_{R_2})\right|}{(E_{R_1,N}-E_{R_1})+(E_{R_2,N}-E_{R_2})}.
$$

\subsection{Ground state potential energy surface of the H$_2$ molecule}

In all our calculations, the $H_2$ molecule lies on the $x$ axis and is centered at the origin. The parameter $R$ is here the interatomic distance in bohrs. 

We numerically observe that $D_N$ is smaller than $S_N$ by a factor of 10 to 100, and that the error cancellation factor $Q_N$ is smaller when the two interatomic distances are close to each other ($R_1 \simeq R_2$). Morevoer, $Q_N$ is almost constant with respect to the cut-off energy $N$.

In Figure~\ref{fig:H2}, we present detailed results for two different pairs of configurations. On the top, the configurations are rather close since the interatomic distances are $R_1=1.464$ and $R_2=1.524$ bohr. For this approximate model, the equilibrium distance is about $R_{\rm eq} \simeq 1.464$ bohrs (the experimental value is $R_{\rm eq}^{\rm exp}\simeq 1.401$ bohrs). The energy difference is better approximated by a factor of about 50 compared to the energies ($Q_N \simeq 0.02$). Moreover the log-log plots of $S_N$ and $D_N$ are almost parallel, which suggests that there is no improvement in the order of convergence when considering energy differences instead of energies; only the prefactor is improved. This is confirmed by the plots of the error cancellation factor $Q_N$, showing that this ratio does not vary much with $N$. 
On the bottom, the configurations are further apart. The interatomic distances are $R_1=1.344$ and $R_2=1.704$ bohrs. We observe a similar behavior except that the error cancellation phenomenon is less pronounced ($Q_N \simeq 0.1$). 

\begin{figure}[ht]
\centering
\begin{tabular}{cc}
 \hspace{-1cm} \includegraphics[width=8cm]{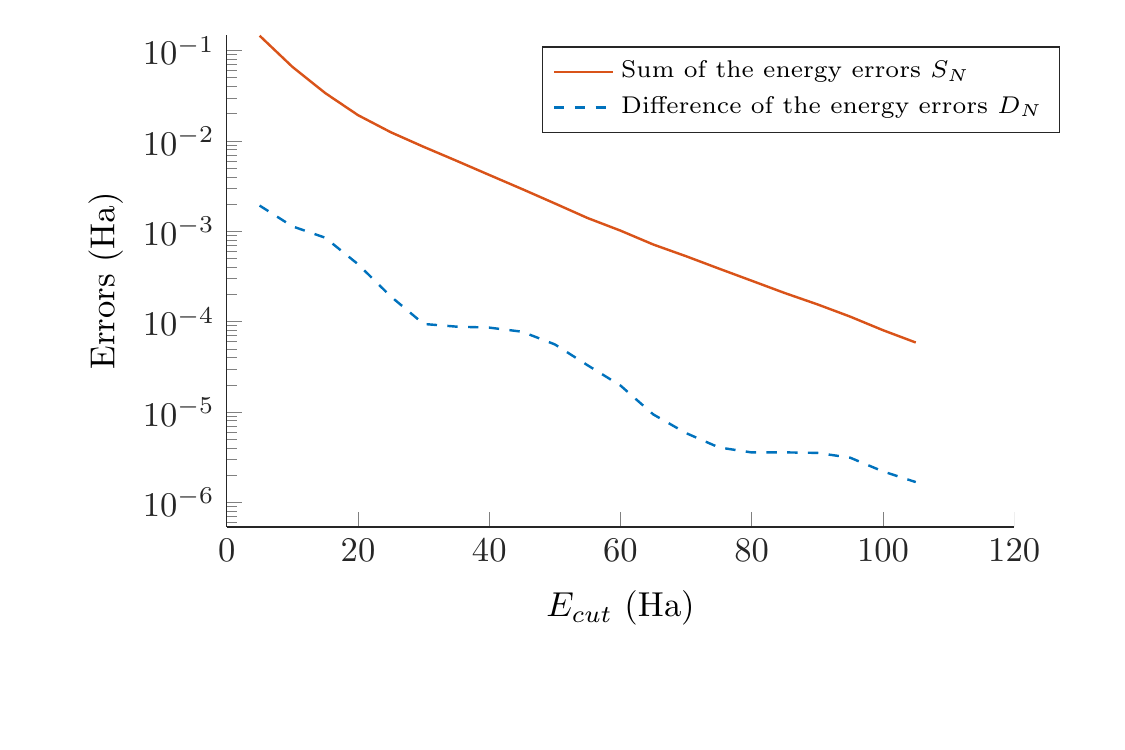}
& \hspace{-1cm} \includegraphics[width=8cm]{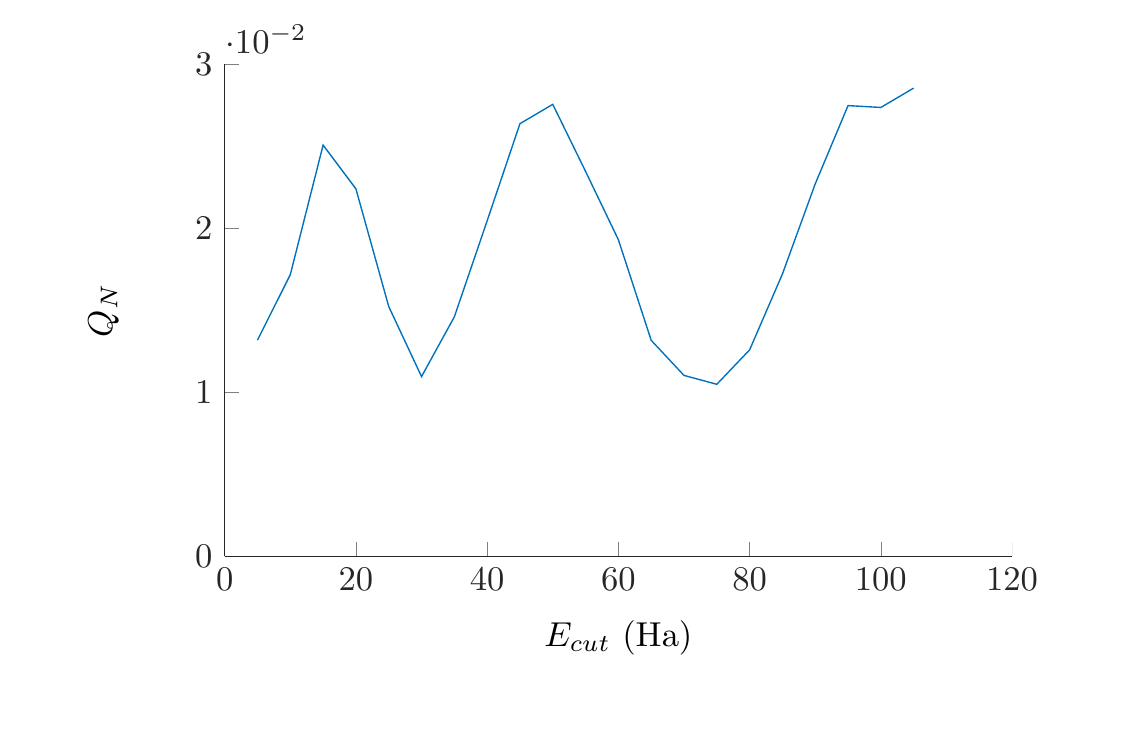} \\ 
 \hspace{-1cm} \includegraphics[width=8cm]{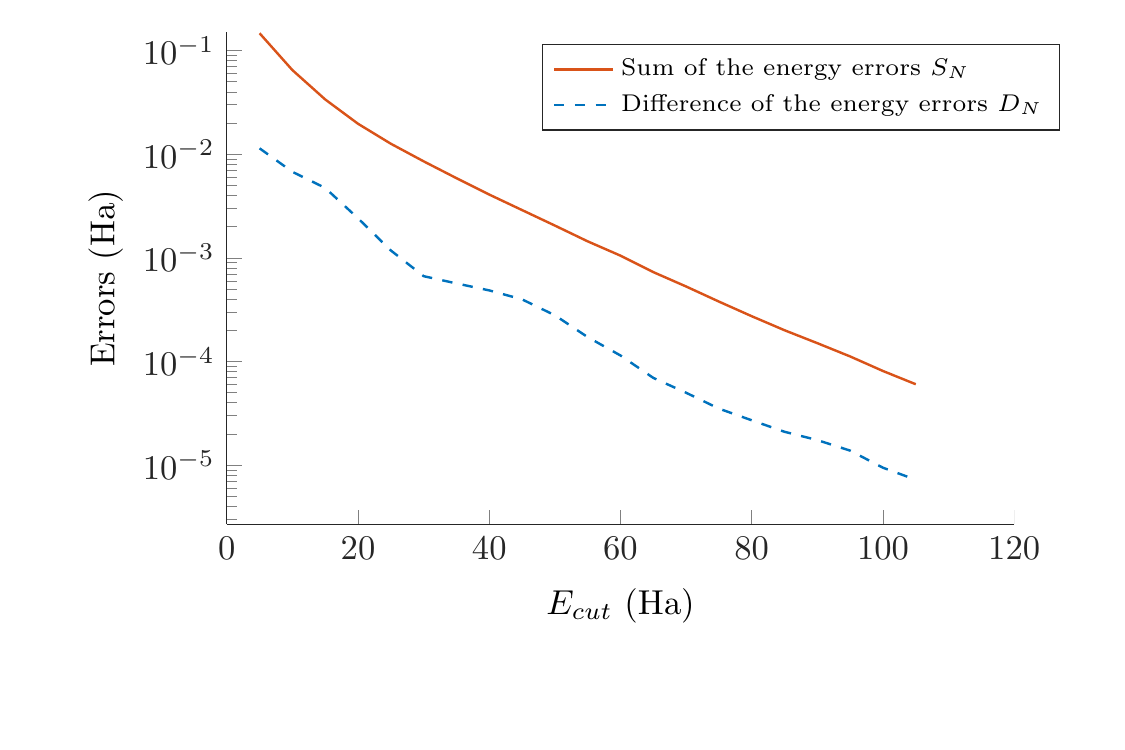} 
 & \hspace{-1cm} \includegraphics[width=8cm]{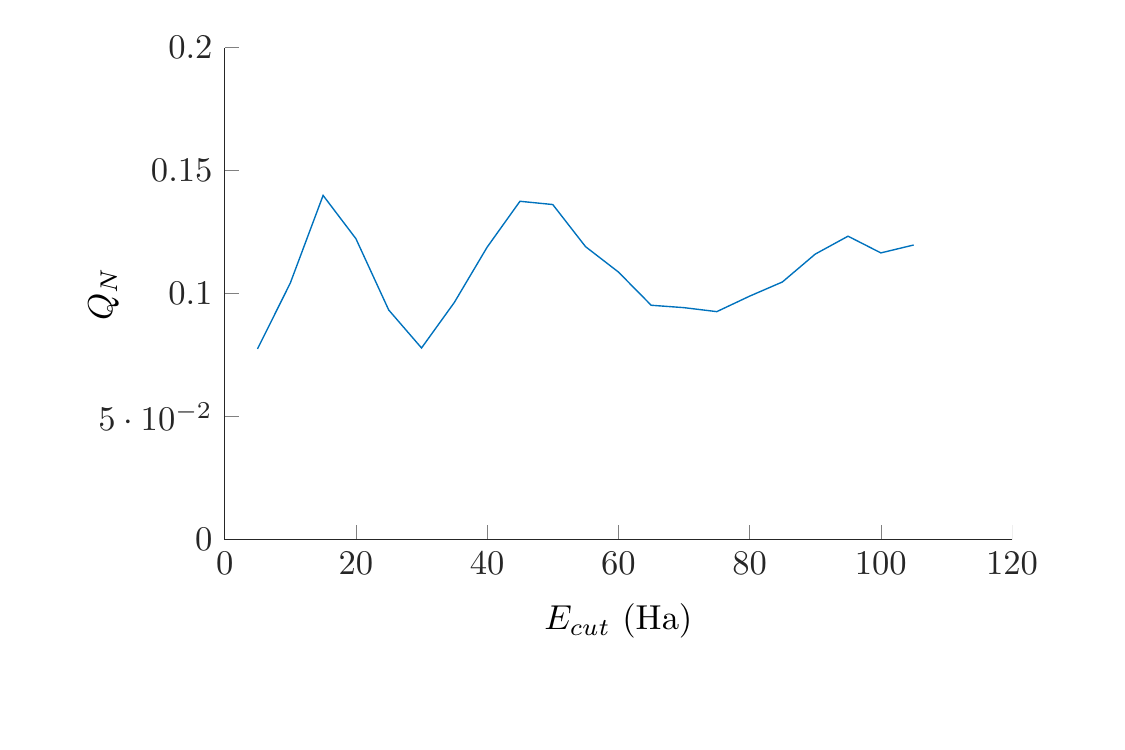}
\end{tabular}
\vspace{-1cm}
\caption{Convergence plots of the quantities $S_N$ and $D_N$ (left) and of the error cancellation factor $Q_N=D_N/S_N$ (right) for two different pairs of interatomic distances for the $H_2$ molecule. Top: $R_1=1.464$ and $R_2=1.524$ bohrs. Bottom: $R_1=1.344$ and $R_2=1.704$ bohrs.}
\label{fig:H2}
\end{figure}

We then compare in Table \ref{tab:tab1} the values of $S_N$ and $D_N$ for different pairs of configurations and for two values of $N=E_{\rm cut}$: a rather coarse energy cut-off $N=30$ Ha, and a quite fine one $N=100$ Ha. One configuration is kept fixed ($R_1=1.284$ bohrs), while the second one varies from $R_2=1.344$ bohrs (close configurations) to $R_2=1.764$ bohrs (distant configurations). 
 We also report, for each pair of configurations, the minimum, maximum, and  mean values of $Q_N$ over the different tested energy cutoffs $5 \le N \le 105$ Ha. We also observe that $Q_N$ increases with $R_2-R_1$ on the range $R_2=[1.344,1.764]$.

\begin{table}[ht]
\begin{tabular}{|c|c||c|c||c|c||c|c|c|}
\hline
$R_1$ & $R_2$ & $S_{N=30}$ &
$D_{N=30}$ & $S_{N=100}$ & $D_{M=100}$  &
$\min(Q_N)$ & $\max(Q_N)$ & $\text{mean}(Q_N)$ \\
\hline
1.284 & 1.344  & 9.410 & 0.1985 & 0.09157 & 0.00112  &
0.0103 & 0.0340  & 0.0212 \\
\hline
1.284  & 1.404  & 9.268 & 0.3408 & 0.08990 & 0.00279  &
0.0216  & 0.0633  & 0.0413  \\
\hline
1.284  & 1.464  & 9.160 & 0.4491 & 0.08772  & 0.00497  &
0.0375  & 0.0895  & 0.0610  \\
\hline
1.284  & 1.524  & 9.065 & 0.5436 & 0.08552  & 0.00717  &
0.0544  & 0.1107  & 0.0802  \\
\hline
1.284  & 1.584  & 8.969 & 0.6394 & 0.08380  & 0.00889  &
0.0713  & 0.1285  & 0.0985  \\
\hline
1.284  & 1.644  & 8.863 & 0.7456 & 0.08274  & 0.00995  &
0.0841  & 0.1455  & 0.1151  \\
\hline
1.284  & 1.704  & 8.744 & 0.8646 & 0.08213  & 0.01056  &
0.0983  & 0.1642  & 0.1302  \\
\hline
1.284  & 1.764  & 8.615 & 0.9937 & 0.08154  & 0.01115  &
0.1072 & 0.1802  & 0.1440 \\
\hline
\end{tabular}
\caption{Comparison of $S_N$, $D_N$ and $Q_N$ for different atomic configurations of the H$_2$ molecule. Distances are in bohrs, energies in mHa.}
\label{tab:tab1}
\end{table}

\subsection{Energy of a simple chemical reaction}

In this section, we consider the energy difference between two very different configurations of a system consisting of two oxygen atoms and four hydrogen atoms. The first configuration, denoted by $R_1$, corresponds to the chemical system 2 H$_2$O (two water molecules) and the second one, denoted by $R_2$, to the chemical system 2 H$_2$ + O$_2$, all these molecules being in their equilibrium geometry (see Figure~\ref{fig:H2O}). The energy difference between the two configurations thus provides a rough estimate of the energy of the chemical reaction 
$$
\mbox{2 H$_2$ + O$_2$} \;  \longrightarrow \; \mbox{2 H$_2$O}.
$$
\begin{figure}[ht]
\centering
\begin{tabular}{cc}
\includegraphics[width=8truecm]{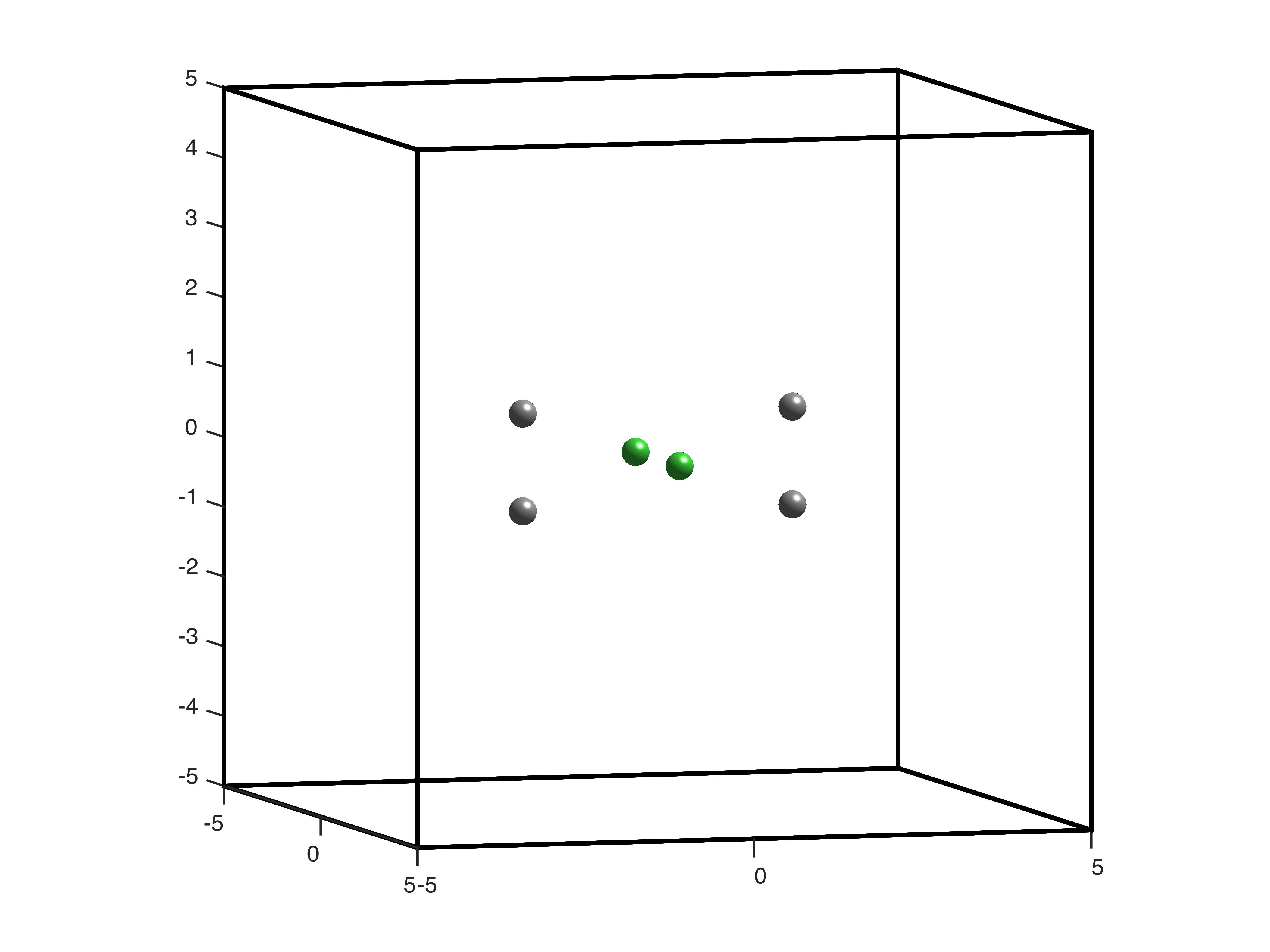} &  \includegraphics[width=8truecm]{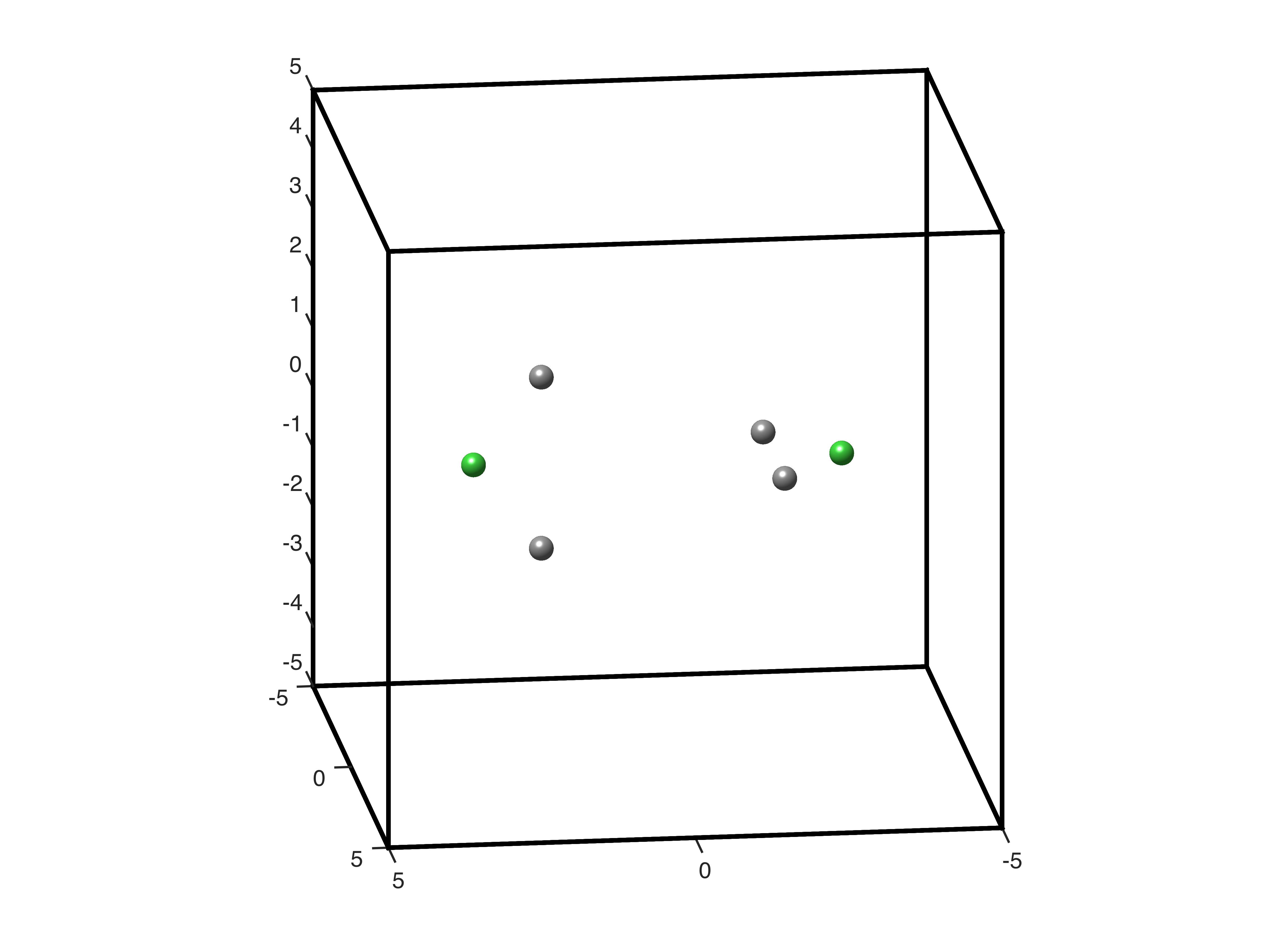}
\end{tabular}
\caption{Graphical representation of the two atomic configurations whose energies are compared. Oxygen atoms are in green, hydrogen atoms in black.}
\label{fig:H2O}
\end{figure}

We can observe on Figure~\ref{fig:H20QN} and Table~\ref{tab:tab11} a similar behavior as for H$_2$, but with a better error cancellation factor ($Q_N \simeq 0.005$).

\begin{figure}[ht]
\centering
\begin{tabular}{cc}
 \hspace{-1cm} \includegraphics[width=8cm]{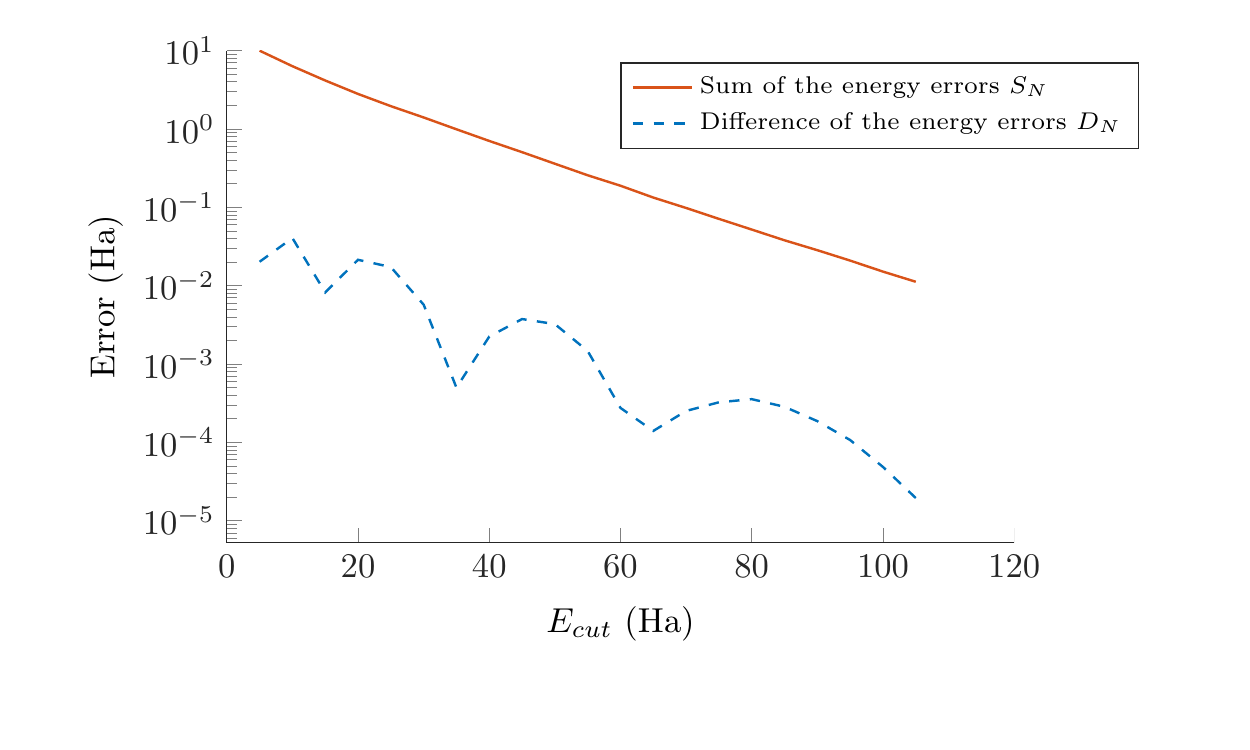}
& \hspace{-1cm} \includegraphics[width=8cm]{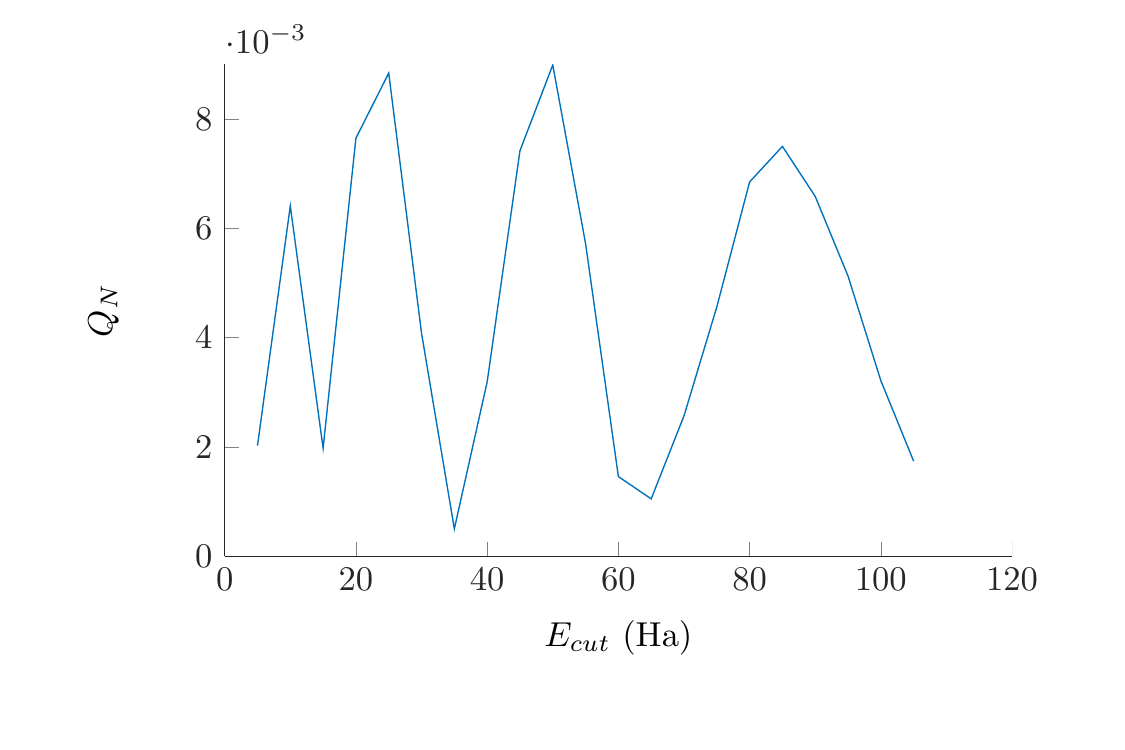}
\end{tabular}
\vspace{-1cm}
\caption{Convergence plots of the quantities $S_N$ and $D_N$ (left) and of the error cancellation factor $Q_N=D_N/S_N$ (right) for the two different configurations displayed on Figure~\ref{fig:H2O}.}
\label{fig:H20QN}
\end{figure}

\begin{table}[ht]
\centering
\begin{tabular}{|c|c||c|c||c|c||c|c|c|}
\hline
$S_{N=30}$ & $D_{N=30}$ & $S_{N=100}$ & $D_{N=100}$  & $\min(Q_N)$ & $\max(Q_N)$ & $\text{mean}(Q_N)$ \\
1403 & 5.726 & 15.12 & 0.0485 & 0.0005036 & 0.008986 & 0.004640 \\
\hline
\end{tabular}
\caption{Comparison of $S_N$, $D_N$ (in mHa) and $Q_N$ for the two different configurations displayed on Figure~\ref{fig:H2O}.}
\label{tab:tab11}
\end{table}

\section{Mathematical analysis of a toy model}
\label{sec:toy_problem}

We now present a simple one-dimensional periodic linear Schr\"odinger model for which the discretization error cancellation phenomenon observed in the previous section can be explained with full mathematical rigor.

We denote by 
\[
	L^2_{\rm per} := \left\{ u\in L^2_{\rm loc}(\R) \; \middle|	 \; u \; \text{is} \;	1-\text{periodic}		 \right\}
\]
the vector space of the 1-periodic locally square integrable real-valued functions on $\R$, and by 
\[
	H^1_{\rm per}:= \left\{	u\in L^2_{\rm per}  \; \middle|	 \; u'\in  L^2_{\rm per} 	\right\}
\]
the associated order-1 Sobolev space. For two given parameters $z_1,z_2>0$, we consider the family of problems, indexed by $R \in (0,1)$, consisting in finding the ground state $(u_R,E_R) \in H^1_{\rm per} \times \mathbb{R}$ of
\begin{equation}
\label{eq:pb_cont}
\left \{
\begin{array}{l}
 \displaystyle  \left( -\frac{d^2}{dx^2} -  \sum_{m \in \Z} z_1 \delta_m - \sum_{m \in \Z} z_2 \delta_{m+R} \right) u_R = E_R u_R, \\
 \displaystyle   \int_0^1 u_R^2(x) dx = 1,  \quad u_R \ge 0, \\
\end{array}
\right.
\end{equation}
where $\delta_a$ denotes the Dirac mass at point $a \in \R$. A variational formulation of the problem is: find the ground state $(u_R,E_R) \in H^1_{\rm per} \times \mathbb{R}$ of 
\begin{equation}
\label{eq:pb_vare}
\left \{
\begin{array}{l}
\displaystyle  	\forall  v\in H^1_{\rm per}, \; \int_0^1 u'_R(x) v'(x) dx - z_1 u_R(0) v(0) - z_2 u_R(R) v(R) = E_R \int_0^1 u_R(x) v(x) dx,  \\
\displaystyle    \int_0^1 u_R^2(x) dx = 1, \quad u_R \ge 0. \\
\end{array}
\right.
\end{equation}
\begin{remark}
The ground state eigenvalue $E_R$ is negative. Indeed, using the variational characterization of the ground state energy, we get
\begin{equation*}
\displaystyle E_R=\min_{\displaystyle v\in H^1_{\rm per} \setminus \left\{0\right\}} \frac{\displaystyle \int_0^1 v'(x)^2 dx - z_1 v(0)^2 - z_2 v(R)^2}{\displaystyle \int_0^1 v^2(x) dx} < 0,
\end{equation*}
since the Rayleigh quotient is equal to $-z_1-z_2 < 0$ for the constant test function $v = 1$.
\end{remark}
Denoting by $k_R=\sqrt{-E_R}$, we have
\begin{equation}
\label{eq:uplus}
\left\{
\begin{array}{ll}
	u_R(x) = A e^{k_R x}+B e^{-k_R x}, & \quad \forall x \in [0,R], \\
	u_R(x) = C e^{k_R x}+D e^{-k_R x}, & \quad \forall x \in [R-1,0), 
\end{array}	
\right.
\end{equation}
where $A$, $B$, $C$, and $D$ are real-valued constants. Since the function $u_R$ is 1-periodic and continuous on $\R$ and its derivative satisfies the jump conditions $u_R'(m+0)-u_R'(m-0) = - z_1 u_R(m)$ and $u_R'(m+R+0)-u_R'(m+R-0) = -z_2 u_R(m+R)$ for all $m \in \Z$, the coefficients $A$, $B$, $C$, $D$ solve the linear system
$$
\underbrace{\left( \begin{array}{cccc}  1 & 1 & -1 & -1 \\ e^{k_R R} & e^{-k_R R} &-  e^{k_R (R-1)} & - e^{-k_R (R-1)} \\
k_R+z_1 & -k_R+z_1 & -k_R & k_R \\ (k_R-z_2) e^{k_R R} & - (k_R+z_2) e^{-k_R R} &-  k_Re^{k_R (R-1)} & k_R e^{-k_R (R-1)}
\end{array} \right)}_{M(k_R)} \left( \begin{array}{c} A \\ B \\ C \\ D  \end{array} \right)  = \left( \begin{array}{c} 0 \\ 0 \\ 0 \\ 0 \end{array} \right).
$$
The wave vector $k_R$ is the lowest positive root of the function $k \mapsto {\rm det}(M(k))$. The coefficients $(A,B,C,D)$ are then uniquely determined by the normalization condition $\|u_R\|_{L^2_{\rm per}}=1$ and the positivity of $u_R$. Exact solutions for two different values of the triplet of parameters $(z_1,z_2,R)$ are plotted in Figure \ref{fig:fig1}.
\begin{figure}[!ht]
\centering
\begin{minipage}{0.47\linewidth}
\hspace*{-1cm}
\includegraphics[scale=0.85]{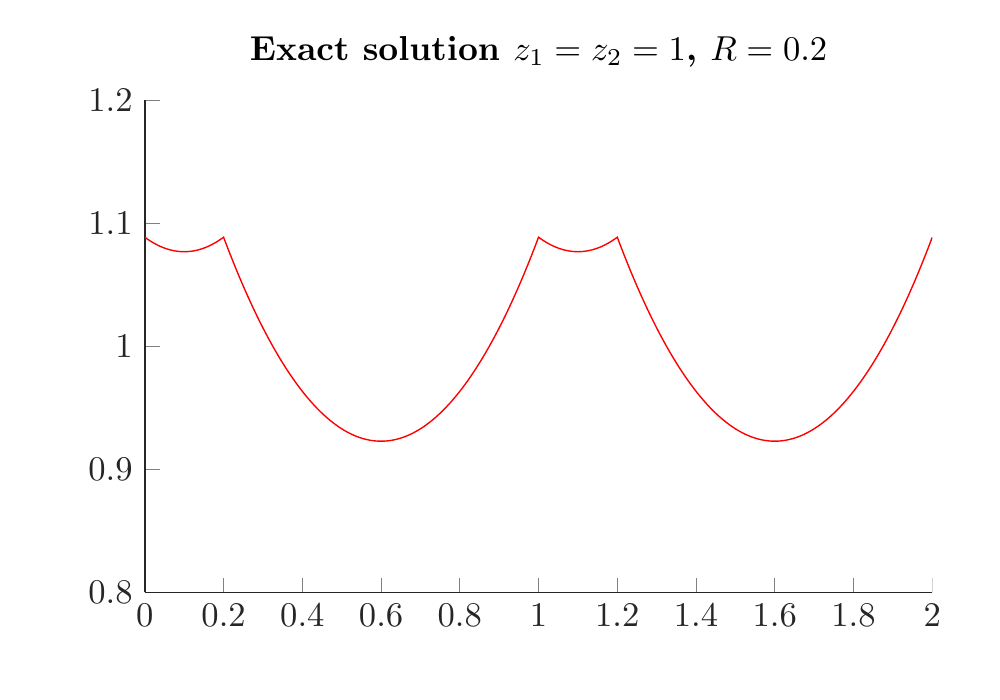}
\end{minipage}
\begin{minipage}{0.47\linewidth}
\hspace*{-1cm}
\includegraphics[scale=0.85]{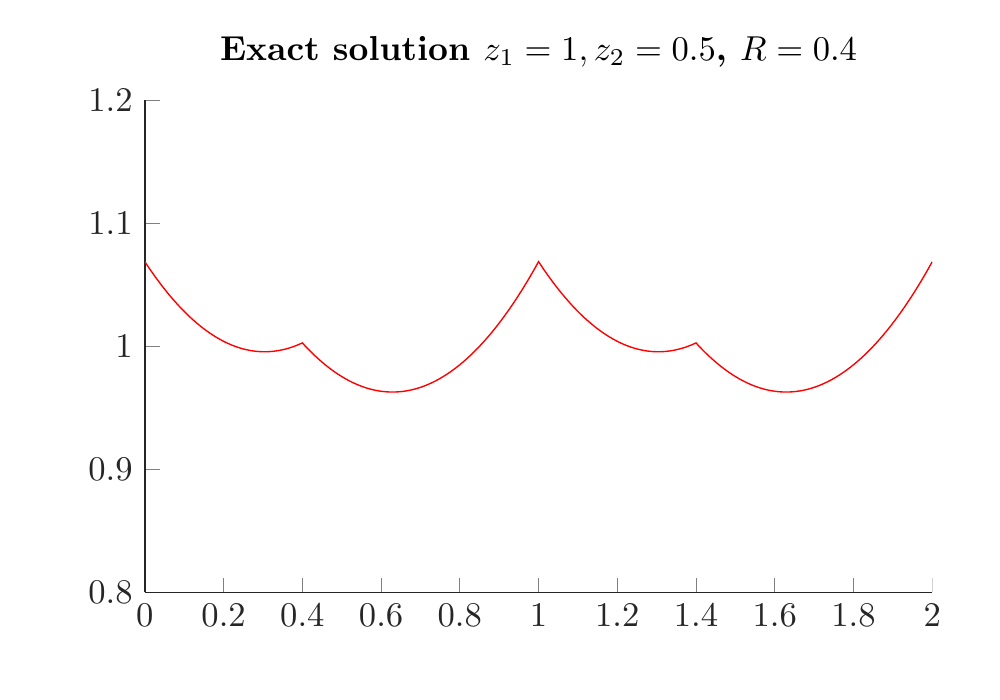}
\end{minipage}
\caption{Plot of the exact solutions of \eqref{eq:pb_cont} for two sets of parameters.}
\label{fig:fig1}
\end{figure}

An approximate solution of the problem is obtained using the PW discretization method. Denoting by 
\[
	X_N := \text{Span} \left\{  v_N(x) = \sum_{k \in \Z, \; |k| \le N} \widehat v_k e^{2\pi i k x} \; \bigg| \; \widehat v_k \in \C, \; \widehat v_{-k}  = \overline{\widehat v_k} \right\}  \subset H^1_{\rm per},
\]
the variational approximation of problem \eqref{eq:pb_vare} in $X_N$ consists in computing the ground state $(u_{R,N},E_{R,N}) \in X_N \times \R$ of 
\begin{equation}
\label{eq:pb_varN}
\left \{
\begin{array}{l}
 \displaystyle \forall v_N\in X_N,\quad \int_0^1 u_{R,N}' v_N' - z_1 u_{R,N}(0) v_N(0) - z_2  u_{R,N}(R) v_N(R) \
	= E_{R,N} \int_0^1 u_{R,N} v_N,  \\
 \displaystyle   \int_0^1 u_{R,N}^2=1, \quad \int_0^1 u_{R,N} \ge 0. \\
\end{array}
\right.	
\end{equation}
The conditions $\widehat v_{-k}  = \overline{\widehat v_k}$ in the definition of $X_N$ is equivalent to imposing that the elements of $X_N$ are real-valued functions. For convenience, the discretization parameter $N$ here corresponds to the cut-off in momentum space. As above, we consider the  error cancellation factor 
\begin{equation}
\label{eq:imp}
\displaystyle	Q_N = \frac{|(E_{R_1,N}-E_{R_2,N})-(E_{R_1}-E_{R_2})|}{(E_{R_1,N}-E_{R_1})+(E_{R_2,N}-E_{R_2})}
\end{equation}
associated with the pair of configurations $(R_1,R_2)$.

Note that imposing the condition $\int_0^1 u_{R,N} \ge 0$, we ensure that the discrete eigenfunction $u_{R,N}$ will approximate the positive eigenfunction $u_R$ to the continuous problem \eqref{eq:pb_cont} and not $-u_R$.
\begin{theorem}[Asymptotic expressions of the energy error and of the error cancellation factor]
\label{th:th1}
For all $z_1,z_2>0$ and $R \in (0,1)$, we have for all $\epsilon >0$,
\begin{equation} \label{eq:Ediff}
E_{R,N}-E_R = \frac{\alpha_R}N - \frac{\alpha_R}{2N^2}  +\frac{\beta^{(1)}_{R,N}}{N} + \frac{\gamma_R}{N} \eta_{R,N}+ o\left( \frac{1}{N^{3-\epsilon}}\right),
\end{equation}
where
$$
\alpha_R:= \frac{z_1^2 u_R(0)^2 + z_2^2 u_R(R)^2}{2\pi^2}, \quad \gamma_R:= \frac{z_1z_2 u_{R}(0)u_R(R)}{\pi^2}, \quad
\eta_{R,N}:= N \sum_{k=N+1}^{+\infty} \frac{\cos(2\pi k R)}{k^2}, 
$$
$$
\beta^{(1)}_{R,N} := \frac{z_1^2 u_R(0) (u_{R,N}(0)-u_R(0)) + z_2^2 u_R(R) (u_{R,N}(R)-u_R(R))}{2\pi^2}.
$$
In addition 
$$
|\eta_{R,N}| \le \min \left( 1, \frac{2+\frac{\pi^3}8}{|\sin(\pi R)| N} \right),
$$
and for all $\epsilon > 0$, there exists $C_\epsilon \in \R_+$ such that
$$
|\beta^{(1)}_{R,N}| \le \frac{C_\epsilon}{N^{1-\epsilon}}.
$$
As a consequence, we have for all $z_1,z_2>0$ and all $R_1,R_2\in (0,1)$,
\begin{equation} \label{eq:CVQN}
	\lim\limits_{N \to +\infty} Q_N =  \frac{|\alpha_{R_1}-\alpha_{R_2}|}{\alpha_{R_1}+\alpha_{R_2}}=\frac{\left| z_1^2 \left( u_{R_1}(0)^2 - u_{R_2}(0)^2 \right) + z_2^2 \left( u_{R_1}(R_1)^2 -  u_{R_2}(R_2)^2 \right)\right|}{z_1^2 (u_{R_1}(0)^2 + u_{R_2}(0)^2 ) + z_2^2  (u_{R_1}(R_1)^2 + u_{R_2}(R_2)^2)} .
\end{equation}
\end{theorem}
The proof of the above theorem is given in Appendix. We deduce from \eqref{eq:Ediff} that the discretization error $E_{R,N}-E_R$ on the energy of the configuration $R$ is the sum of 
\begin{enumerate}
\item a leading term $\alpha_R N^{-1}$ of order 1 (in $N^{-1}$);
\item three terms $- 1/2 \alpha_RN^{-2}$, $\beta^{(1)}_{R,N}N^{-1}$, and  $\gamma_{R}N^{-1} \eta_{R,N}$ which are roughly of order 2;
\item higher order terms which are roughly of order $3$ and above.
\end{enumerate}
The leading term $\alpha_R N^{-1}$ has a very simple expression and the prefactor $\alpha_R$ does not vary much with respect to~$R$ (see Figure~\ref{fig:alphaR}). This explains the phenomenon of discretization error cancellation. Regarding the second order corrections on $E_{R,N}-E_R$, we have observed numerically (see Figure~\ref{fig:resnum}) that
\begin{itemize}
\item  the terms $-\frac 12 \alpha_RN^{-2}$ and $\gamma_{R}N^{-1} \eta_{R,N}$ are of about the same order of magnitude in absolute values, that the former is always negative (since $\alpha_R >0$), but that the latter can be either positive or negative, so that the sum of these two contributions can be either significant or negligible;
\item the term $\beta^{(1)}_{R,N}N^{-1}$ is smaller in absolute value than the other two terms, and seems to be always negative. Our numerical calculations indeed show that $u_{R,N}(0) < u_R(0)$ and $u_{R,N}(R) < u_R(R)$, which is not very surprising since the function $u_R$ has cusps at points $x=0$ and $x=R$ (see Figure~\ref{fig:fig1}). These inequalities have not been rigorously established though.
\end{itemize}

\begin{figure}[ht]
\centering
\includegraphics[scale=1]{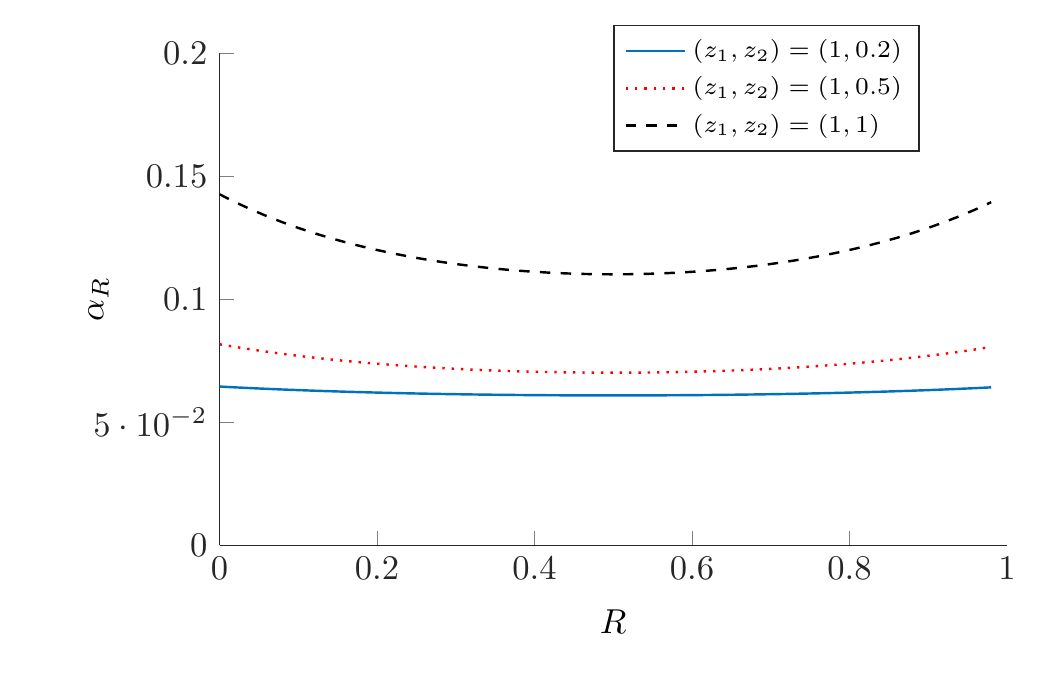}
\caption{Plots of the function $R \mapsto \alpha_R$ for three sets of parameters $(z_1,z_2)$.}
\label{fig:alphaR}
\end{figure}

\begin{figure}[ht]
\centering
\begin{tabular}{cc}
 \hspace{-1cm} \includegraphics[width=9cm]{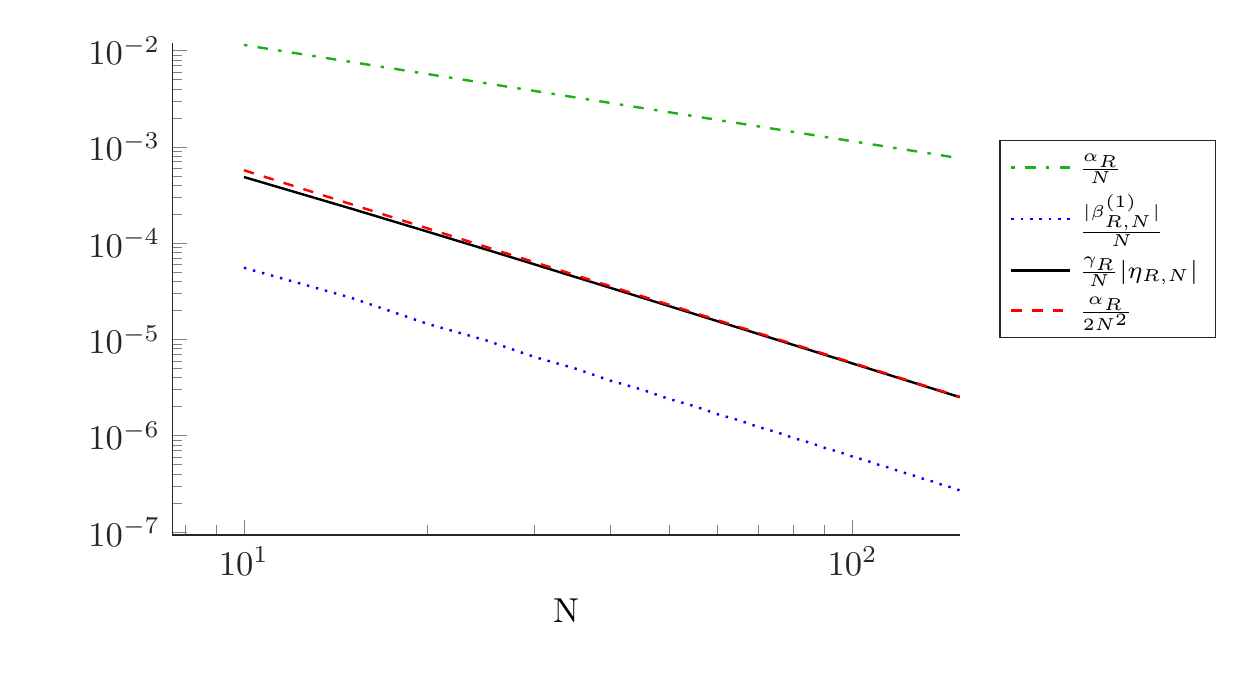}
& \hspace{-1cm} \includegraphics[width=9cm]{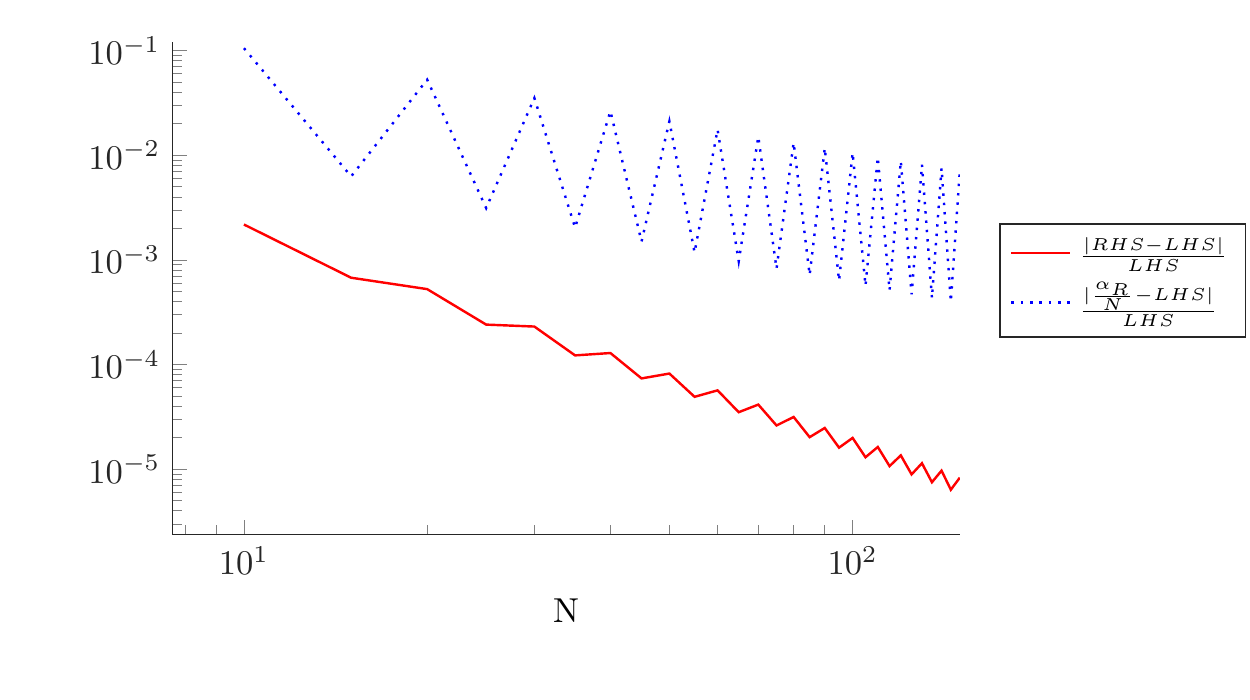} \\ 
 \hspace{-1cm} \includegraphics[width=9cm]{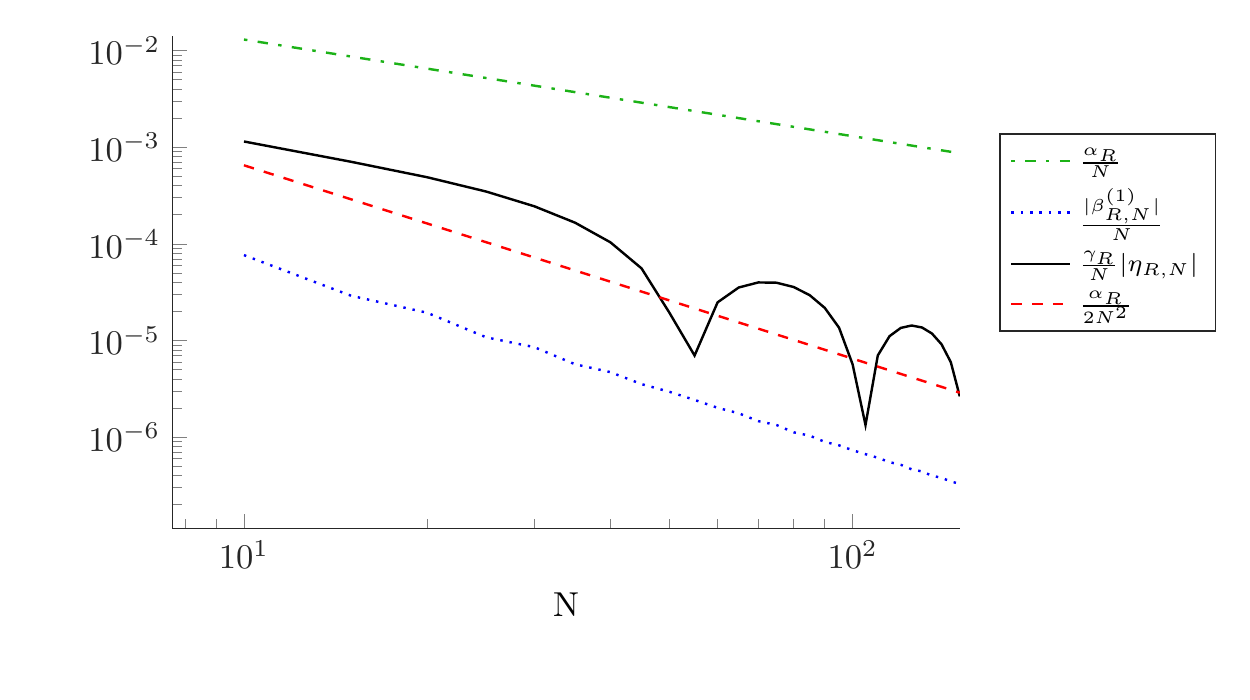} 
 & \hspace{-1cm} \includegraphics[width=9cm]{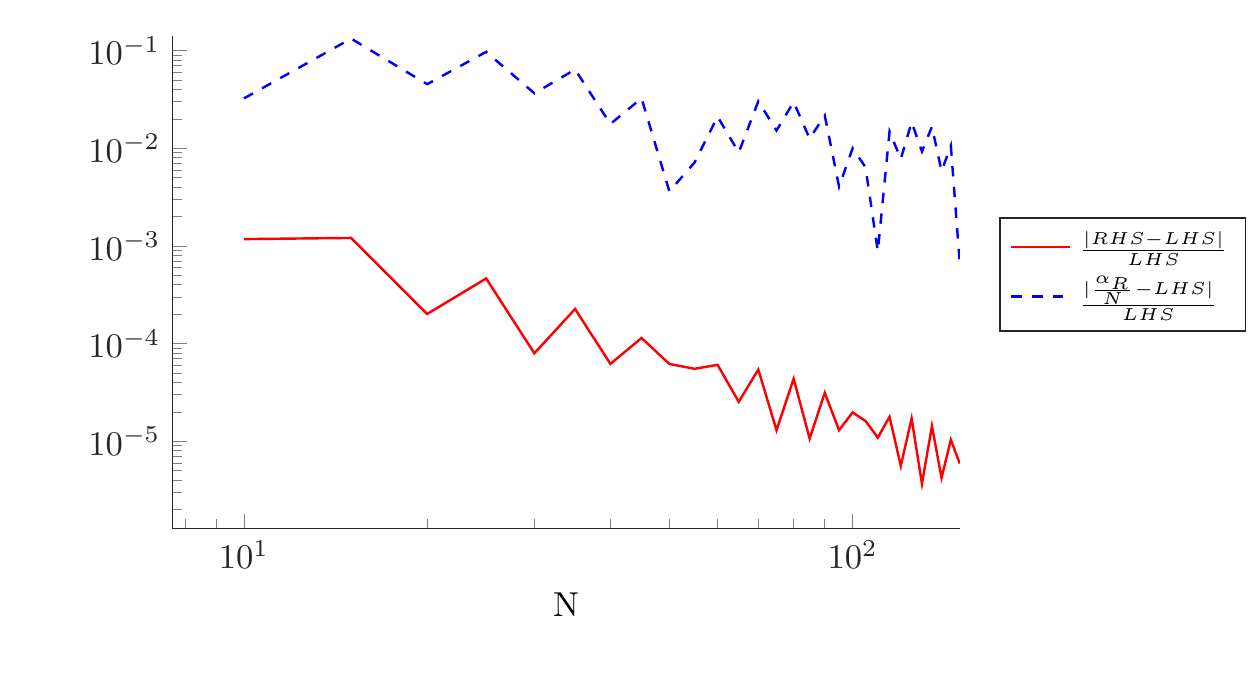}
\end{tabular}
\vspace{-1cm}
\caption{Convergence plots of the four quantities  $\frac{\alpha_R}N$,  $\frac{\alpha_R}{2N^2}$, $\frac{|\beta^{(1)}_{R,N}|}{N}$, and $\frac{\gamma_R}{N} |\eta_{R,N}|$ (left) and plots of 
$\frac{ |(\frac{\alpha_R}N - \frac{\alpha_R}{2N^2}  +\frac{\beta^{(1)}_{R,N}}{N} + \frac{\gamma_R}{N} \eta_{R,N})- (E_{R,N}-E_R)|}{E_{R,N}-E_R}$ and
$\frac{ |\frac{\alpha_R}N - (E_{R,N}-E_R)|}{E_{R,N}-E_R}$ (right).
Top: $z_1=z_2=1, R = 0.3$. Bottom: $z_1=z_2=1, R = 0.09$.
}
\label{fig:resnum}
\end{figure}

\begin{figure}[ht]
\centering
\includegraphics[scale=1]{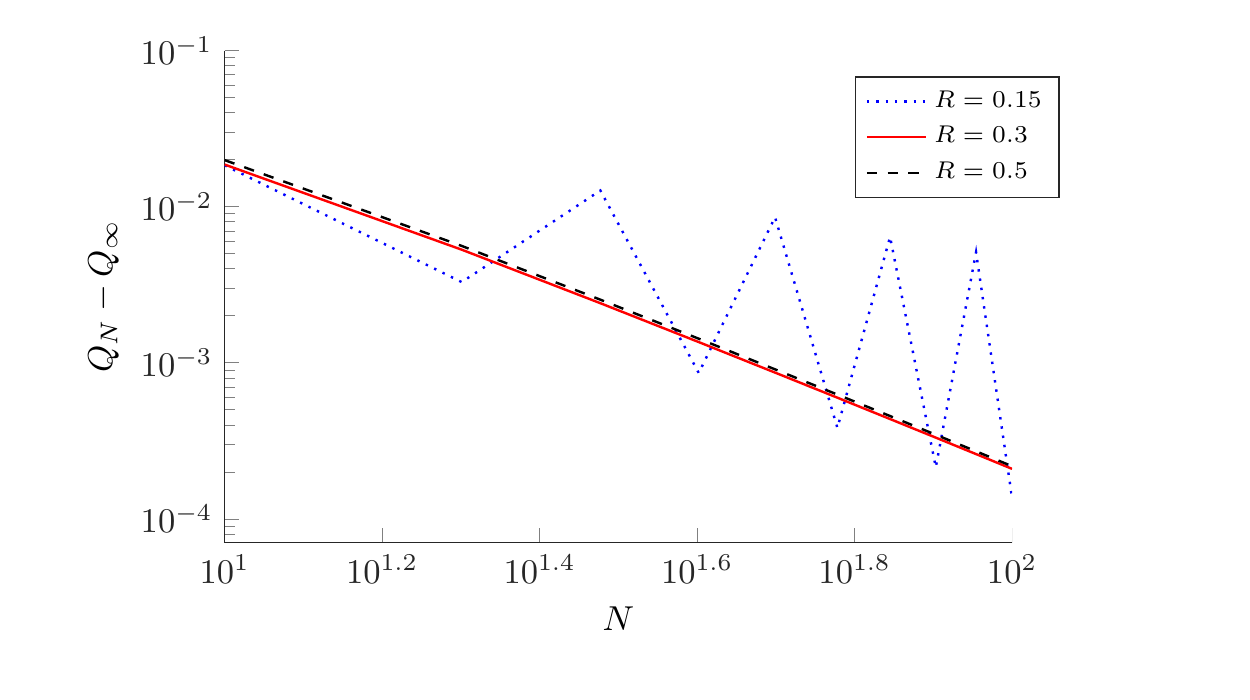}
\caption{Plot of $Q_N-Q_\infty$ for three values of $R$.}
\label{fig:convQN}
\end{figure}
Finally, we observe on Figure \ref{fig:convQN} that $Q_N$ converges to the asymptotic value $Q_\infty$ when $N$ goes to infinity very smoothly for large values of $R$, and with oscillations when $R$ becomes close to zero. Moreover, $Q_N-Q_\infty$ is of order $N^{-2}$.

\section{Appendix: proof of Theorem \ref{th:th1}}

In the sequel, $z_1$ and $z_2$ are fixed positive real numbers. We endow the functional spaces $L^2_{\rm per}$ and $H^1_{\rm per}$ with their usual scalar products
$$
\langle u|v \rangle_{L^2_{\rm per}} := \int_0^1 u(x) v(x) \, dx \quad \mbox{and} \quad \langle u|v \rangle_{H^1_{\rm per}} := \langle u|v \rangle_{L^2_{\rm per}} + \langle u'|v' \rangle_{L^2_{\rm per}}.
$$
More generally, we endow the Sobolev space
$$
H^s_{\rm per} := \left\{v(x) =\sum_{k \in \Z} \widehat v_k e^{2i\pi kx} \; \bigg|  \; \widehat v_k \in \C, \;  \widehat v_{-k} =\overline{\widehat v_k}, \; \sum_{k \in \Z} (1+(2\pi k)^2)^s |\widehat v_k|^2 < \infty \right\}, 
$$
$s \in \R$, with the scalar product defined by
$$
 \langle u|v \rangle_{H^s_{\rm per}} := \sum_{k \in \Z} (1+(2\pi k)^2)^s \, \overline{\widehat u_k} \, \widehat v_k.
$$
Note that the above two definitions of $\langle u|v \rangle_{H^1_{\rm per}}$ coincide and that $H^0_{\rm per}=L^2_{\rm per}$.  We also denote by $\Pi_N$ the orthogonal projection on $X_N$ for the $L^2_{\rm per}$  (and also $H^s_{\rm per}$) scalar product and by $\Pi_N^\perp = 1 - \Pi_N$.

\medskip

We first recall some useful results on the convergence of $(u_{R,N},E_{R,N})$ to $(u_R,E_R)$. 

\begin{lemma}
\label{lem:lem0}
 Let $R\in (0,1)$. Let $(u_R,E_R)$ be the ground state of the continuous problem \eqref{eq:pb_vare}, and $(u_{N,R},E_{R,N})$ be a ground state of the discretized problem \eqref{eq:pb_varN}. Then, for all $\epsilon > 0$ and all $0 \le s < 3/2$, there exists $C_{s,\epsilon} \in \R_+$ such that
 \begin{equation}\label{eq:uN-u}
 \|u_{R,N}-u_R\|_{H^s_{\rm per}} \le \frac{C_{s,\epsilon}}{N^{3/2-s-\epsilon}}.
 \end{equation}
 In addition, there exist $0 < c \le C < \infty$ such that
  \begin{equation}\label{eq:EN-E}
 c \|u_{R,N}-u_R\|_{H^1_{\rm per}}^2 \le E_{R,N}-E_R \le C \|u_{R,N}-u_R\|_{H^1_{\rm per}}^2 ,
  \end{equation}
 and for all $\epsilon > 0$, there exists  $C_\epsilon \in \R_+$ such that
  \begin{equation}\label{eq:uN(0)-u(0)}
  |u_{R,N}(0)-u_R(0)|+ |u_{R,N}(R)-u_R(R)| \le \frac{C_\epsilon}{N^{1-\epsilon}}.
  \end{equation}
 \end{lemma}

\begin{proof}
We denote by $C^0_{\rm per}$ the space of continuous $1$-periodic functions from $\R$ to $\R$ endowed with the norm defined by
$$
\forall u \in C^0_{\rm per}, \quad \|u\|_{C^0_{\rm per}} :=\max_{x \in \R}|u(x)|.
$$
Recall that $H^s_{\rm per}$ is continuously embedded in $C^0_{\rm per}$ for all $s > 1/2$. In particular, $H^1_{\rm per} \hookrightarrow C^0_{\rm per}$ and there exists $K \in \R_+$ such that
\begin{equation}\label{eq:boundC0}
\forall u \in H^1_{\rm per}, \quad \|u\|_{C^0_{\rm per}} \le  K \|u\|_{H^{3/4}_{\rm per}} \le K \|u\|_{H^1_{\rm per}}^{3/4} \|u\|_{L^2_{\rm per}}^{1/4}.
\end{equation}
In particular, the bilinear form 
$$
\forall (u,v) \in H^1_{\rm per} \times H^1_{\rm per}, \quad a_R(u,v)=\int_0^1 u'v' - z_1u(0)v(0) - z_2u(R)v(R)
$$
is well-defined, symmetric, and continuous on $H^1_{\rm per} \times H^1_{\rm per}$, and we have
\begin{align*}
\forall u \in H^1_{\rm per}, \quad a_R(u,u) &\ge \|u\|_{H^1_{\rm per}}^2 - (z_1+z_2) K^2 \|u\|_{H^1_{\rm per}}^{3/2} \|u\|_{L^2_{\rm per}}^{1/2} - \|u\|_{L^2_{\rm per}}^2 \\ &\ge \frac 12 \|u\|_{H^1_{\rm per}}^2 - \left( 1+\frac{27}{32}(z_1+z_2)^4K^8 \right) \|u\|_{L^2_{\rm per}}^2, 
\end{align*}
using Young's inequality.
The quadratic form $H^1_{\rm per} \ni u \mapsto a_R(u,u) \in \R$ therefore is bounded below and closed. We denote by $H_R$ the unique self-adjoint operator on $L^2_{\rm per}$ associated to $a_R(\cdot,\cdot)$ (see e.g.~\cite[Theorem VIII.15]{ReedSimon1}). Formally,
$$
H_R = - \frac{d^2}{dx^2} - z_1 \sum_{m \in \Z} \delta_m - z_2 \sum_{m \in \Z} \delta_{m+R}.
$$
The domain of $H_R$ being a subspace of $H^1_{\rm per}$, which is itself compactly embedded in $L^2_{\rm per}$, the spectrum of $H_R$ is purely discrete: it consists of an increasing sequence of eigenvalues of finite multiplicities going to $+\infty$. It is easily seen that its ground state eigenvalue $E_R$ is simple. Let us denote by $\mu_R > 0$ the gap between the lowest two eigenvalues of $H_R$. A classical calculation shows that
\begin{align*}
E_{R,N}-E_R &= a_R(u_{R,N}-u_R,u_{R,N}-u_R)-E_R \|u_{R,N}-u_R\|^2_{L^2_{\rm per}} \\
&= \langle u_{R,N} | H_R | u_{R,N}\rangle - E_R.
\end{align*}
First, since $E_R < 0$, we have
$$
E_{R,N}-E_R \le a_R(u_{R,N}-u_R,u_{R,N}-u_R) \le M_R \|u_{R,N}-u_R\|_{H^1_{\rm per}}^2,
$$
where $M_R$ is the continuity constant of $a_R$, which proves the second inequality in \eqref{eq:EN-E}. Second, since $\|u_R\|_{L^2_{\rm per}}=\|u_{R,N}\|_{L^2_{\rm per}}=1$, we have on the one hand
\begin{align*}
	 E_{R,N}-E_R  & = \langle u_{R,N} | H_R | u_{R,N}\rangle - E_R \ge \left( E_R |\langle u_{R,N} |u_R\rangle_{L^2_{\rm per}}|^2 +(E_R+ \mu_R) \left(1- |\langle u_{R,N} |u_R\rangle_{L^2_{\rm per}}|^2 \right) \right) - E_R  \\
	 &= \mu_R \left(1- |\langle u_{R,N} |u_R\rangle_{L^2_{\rm per}}|^2 \right) \ge \mu_R \left(1- \langle u_{R,N} |u_R\rangle_{L^2_{\rm per}} \right)  = \frac{\mu_R}2 \|u_{R,N}-u_R\|^2_{L^2_{\rm per}},
\end{align*}
and, on the other hand,
$$
E_{R,N}-E_R \ge \frac 12  \|u_{R,N}-u_R\|_{H^1_{\rm per}}^2 -  \left( 1+\frac{27}{32}(z_1+z_2)^4K^8 +E_R \right) \|u_{R,N}-u_R\|_{L^2_{\rm per}}^2.
$$
Combining the above two inequalities yields the first inequality in \eqref{eq:EN-E}. Hence, \eqref{eq:EN-E} is proved. 

We deduce from the min-max principle that for each $v_N \in X_N$ such that $\|v_N\|_{L^2_{\rm per}}=1$, we have 
\begin{align*}
E_{R,N}-E_R & \le a_R(v_N,v_N)-E_R = a_R(v_N-u_R,v_N-u_R) - E_R  \|v_N-u_R\|^2_{L^2_{\rm per}} \\ 
& \le \left(M_R-E_R\right) \|v_N-u_R\|_{H^1_{\rm per}}^2.
\end{align*}
Since $z_1 \sum_{m \in \Z} \delta_m+z_2 \sum_{m \in \Z} \delta_{m+R} \in H^{-1/2-\epsilon}_{\rm per}$ for all $\epsilon>0$, we have that $u_R \in H^{3/2-\epsilon}_{\rm per}$. Applying the above estimate to $v_N=\|\Pi_Nu_R\|_{L^2_{\rm per}}^{-1} \Pi_Nu_R$, we get
$E_{R,N}-E_R \le \frac{C_\epsilon}{N^{1-\epsilon}}$. Combining with \eqref{eq:EN-E}, we obtain \eqref{eq:uN-u} for $s=1$. Together with \eqref{eq:boundC0}, this implies in addition that $(u_{R,N})_{N \in \N}$ converges to $u_R$ in $C^0_{\rm per}$. Since 
$$
-u_{R,N}''= z_1 u_{R,N}(0) \Pi_N \left( \sum_{k \in \Z} \delta_m \right) + z_2 u_{R,N}(R) \Pi_N \left( \sum_{k \in \Z} \delta_{m+R} \right) + E_{R,N} u_{R,N},
$$
and the right hand-side converges to $-u_R''$ in $H^{-1/2-\epsilon}_{\rm per}$ for all $\epsilon > 0$, the sequence $(u_{R,N})_{N \in \N}$ converges to $u_R$ in $H^{3/2-\epsilon}_{\rm per}$ for all $\epsilon > 0$. By interpolation, we then obtain \eqref{eq:uN-u} for all $1 \le s < 3/2$. We finally obtain \eqref{eq:uN-u} for $s=0$ by a classical Aubin-Nitsche argument, and we conclude by interpolation that the result also holds true for all $0 \le s < 1$.

To prove \eqref{eq:uN(0)-u(0)}, we infer from the Sobolev embedding $H^{1/2+\epsilon}_{\rm per} \hookrightarrow C^0_{\rm per}$, that
$$
 |u_{R,N}(0)-u_R(0)|+ |u_{R,N}(R)-u_R(R)| \le 2 \|u_{R,N}-u_R\|_{C^0_{\rm per}} \le 2C'_\epsilon \|u_{R,N}-u_R\|_{H^{1/2+\epsilon}_{\rm per}},
$$
and we conclude using \eqref{eq:uN-u} with $s=1/2+\epsilon$.
\end{proof}

The following lemma provides an expression of the leading term of the energy difference $E_{R,N}-E_R$.
\begin{lemma}
\label{lem:lem1}
Let $z_1,z_2>0$. Let $R\in (0,1)$. Let $(u_R,E_R)$ be the ground state of the continuous problem \eqref{eq:pb_vare}, and $(u_{R,N},E_{R,N})$ be a ground state of the discretized problem \eqref{eq:pb_varN}. Then, for all $\epsilon > 0$, 
\begin{equation}
\label{eq:diff_vp}
	E_{R,N}-E_R = z_1 u_{R,N}(0)(\Pi_N^\perp u_R) (0) + z_2 u_{R,N}(R)(\Pi_N^\perp u_R) (R) + o\left( \frac{1}{N^{3-\epsilon}}\right),
\end{equation}
when $N$ goes to $+\infty$.
\end{lemma}

\begin{proof}
The variational formulation \eqref{eq:pb_vare} with $v = u_{R,N}$ gives
\[
	E_R \int_0^1 u_{R,N} u_R = \int_0^1 u_{R,N}' u_R' - z_1 u_{R,N}(0)  u_R(0) - z_2 u_{R,N}(R) u_R(R).
\]
The variational formulation \eqref{eq:pb_varN} with $v_N = \Pi_N u_R$ gives
\[
	E_{R,N} \int_0^1 u_{R,N} (\Pi_N u_R) = \int_0^1 u_{R,N}' (\Pi_Nu_R)' - z_1  u_{R,N}(0) (\Pi_N u_R)(0) - z_2 u_{R,N}(R) (\Pi_N u_R)(R).	
\]
Subtracting these two equalities, and noting first that $\displaystyle \int_0^1 u_{R,N} (\Pi_N u_R) = \int_0^1 u_{R,N} u_R$, and second that $\displaystyle \int_0^1 u_{R,N}' (\Pi_Nu_R)' = \int_0^1 u_{R,N}' u_R'$, since $u_{R,N}\in X_N$ and the orthogonal projection $\Pi_N$ and the derivation commute, we get
\[
	(E_{R,N}-E_R)  \int_0^1 u_{R,N} u_R = z_1 u_{R,N}(0) (\Pi_N^\perp u_R)(0) + z_2 u_{R,N}(R) (\Pi_N^\perp u_R)(R).
\]
Moreover, since $\displaystyle \int_0^1 u_R^2 = \int_0^1 u_{R,N}^2 =1$, we have
\[
	\int_0^1 u_{R,N} u_R = 1 - \frac{1}{2}\int u_R^2 - \frac{1}{2}  \int_0^1 {u_{R,N}}^2 + \int_0^1 u_{R,N} u_R = 1 - \frac{1}{2} \|u_{R,N} - u_R\|_{L^2_{\rm per}}^2.
\]
Hence,
\[
	(E_{R,N}-E_R) \left(1 - \frac{1}{2} \|u_{R,N} - u_R\|_{L^2_{\rm per}}^2 \right)  = z_1 u_{R,N}(0) (\Pi_N^\perp u_R)(0) +  z_2 u_{R,N}(R) (\Pi_N^\perp u_R)(R).
\]
Using estimates \eqref{eq:uN-u} for $s=0$ and \eqref{eq:EN-E}, we obtain that for all $\epsilon > 0$,
\[
\displaystyle	1 - \frac{1}{2} \|u_{R,N} - u_R\|_{L^2_{\rm per}}^2 =1 + o\left( \frac{1}{N^{3-\epsilon}}\right), \quad \mbox{when $N \to +\infty$}. 
\]
This concludes the proof of Lemma~\ref{lem:lem1}.
\end{proof}

The following lemma provides an explicit expression of the quantities $(\Pi_N^\perp u_R)(0)$ and $(\Pi_N^\perp u_R)(R)$ appearing in \eqref{eq:diff_vp}.

\begin{lemma} Let $z_1,z_2>0$. For all $R \in (0,1)$, all $N \in \N$, and all $x \in \R$,
\label{lem:lem2}
\begin{equation}\label{eq:PiNuNx}
(\Pi_N^\perp u_R) (x) =  \sum_{k =N+1}^{+\infty}   \frac{2}{k_R^2+4\pi^2 k^2} \left( z_1 u_R(0) \cos (2\pi kx) + z_2 u_R(R) \cos(2\pi k(x-R)) \right).
\end{equation}
\end{lemma}

\begin{proof}
In order to estimate $(\Pi_N^\perp u_R)(x)$, we first need to compute the Fourier coefficients of $u_R$
\begin{equation}
\label{eq:cka}
\forall k \in \Z, \quad \widehat{u_R}(k) := \int_0^1 u_R(x) e^{-2i\pi k x} \, dx.
\end{equation}
Using the periodicity of $u_R$, we can rewrite the first equation in \eqref{eq:pb_cont} as
$$
-u_R'' - z_1 u_R(0) \left( \sum_{m \in \Z} \delta_m \right) - z_2 u_R(R) \left( \sum_{m \in \Z} \delta_{m+Z} \right) = E_R u_R.
$$
Taking the Fourier transform, and using the relation $E_R=-k_R^2$, we obtain
$$
4\pi^2 k^2 \widehat{u_R}(k)	- z_1 u_R(0) -z_2 u_R(R) e^{-2i\pi kR} = -k_R^2 \widehat{u_R}(k).
$$
Hence, for all $k\in \Z$,
\begin{equation}
\label{eq:cka2}
	\widehat{u_R}(k)	=   \frac{1}{k_R^2+4\pi^2 k^2} \left( z_1 u_R(0) + z_2 u_R(R) e^{-2i \pi k R} \right).
\end{equation}
Consequently,
\begin{align*}
(\Pi_N^\perp u_R) (x) &= \sum_{k \in \Z, \; |k| > N} \widehat{u_R}(k) e^{2i\pi k x} = \sum_{k \in \Z, \; |k| > N} \frac{1}{k_R^2+4\pi^2 k^2} \left( z_1 u_R(0) + z_2 u_R(R) e^{-2i \pi k R} \right) e^{2i\pi k x} \\
&=  \sum_{k =N+1}^{+\infty}   \frac{2}{k_R^2+4\pi^2 k^2} \left( z_1 u_R(0) \cos (2\pi kx) + z_2 u_R(R) \cos(2\pi k(x-R)) \right),
\end{align*}
which completes the proof of Lemma~\ref{lem:lem2}.
\end{proof}

The last technical lemma we need provides an estimates of the series in \eqref{eq:PiNuNx} for $x=0$ and $x=R$. 

\begin{lemma} 
\label{lem:lem3}
Let $\R \ni R \mapsto k_R \in \R$ be a positive bounded function and $M=\sup_{R \in \R} k_R^2$. We denote by 
$$
f_N(R):=\sum_{k =N+1}^{+\infty} \frac{1}{k_R^2+4\pi^2k^2} \quad \mbox{and} \quad g_N(R):=\sum_{k =N+1}^{+\infty} \frac{\cos(2\pi k R)}{k_R^2+4\pi^2k^2}.
$$
For all $R \in \R \setminus \Z$ we have 
\begin{equation}\label{eq:bound_uN}
f_N(R) = \frac{1}{4\pi^2 N}  a_N+ \phi_N(R), \; \mbox{ with }  \;   a_N =  N \sum_{k=N+1}^{+\infty} \frac{1}{k^2}, \; \; |\phi_N(R)| \le \frac{M}{48\pi^4N^3},
\end{equation}
and 
\begin{equation}\label{eq:bound_vN} 
g_N(R)= \frac{1}{4\pi^2 N}\eta_{N,R} + \psi_N(R),  \; \mbox{ with }  \; \eta_{N,R} = N \sum_{k=N+1}^{+\infty} \frac{\cos(2\pi k R)}{k^2}, \;\; \left| \psi_N(R) \right| \le \frac{M}{48\pi^4N^3}.
\end{equation}
Besides, 
\begin{equation}\label{eq:aNetaRN}
a_N = 1 + \frac{1}{2N} + O\left(\frac{1}{N^2}\right) \qquad \mbox{and} \qquad \left| \eta_{N,R} \right| \le \min\left( 1, \frac{2+\frac{\pi^3}8}{|\sin (\pi R)|N} \right).
\end{equation}
\end{lemma}

\begin{proof} The function $f_N$ can be decomposed as 
$$
f_N(R) = \frac{1}{4\pi^2 N}  a_N+ \phi_N(R),
$$
where 
$$
\phi_N(R) = f_N(R)-\frac{1}{4\pi^2 N}  a_N=-\frac{k_R^2}{4\pi^2} \sum_{k=N+1}^{+\infty} \frac{1}{k^2(k_R^2+4\pi^2 k^2)}.
$$
We have on the one hand
$$
a_N=1 + N \sum_{k=N+1}^{+\infty} \left( \frac{1}{k^2} - \int_{k-1}^k \frac{dt}{t^2} \right) = 1 + N \sum_{k=N+1}^{+\infty} \frac{1}{k^2} \int_0^1 \left( 1- \left(1-\frac sk \right)^{-2} \right) \, ds = 1 + \frac{1}{2N} + O\left(\frac{1}{N^2}\right),
$$
and on the other hand, by a sum-integral comparison,
$$
|\phi_N(R)| \le \frac{M}{4\pi^2} \sum_{k=N+1}^{+\infty} \frac{1}{4\pi^2k^4} \le \frac{M}{48\pi^4N^3}.
$$

Thus, (\ref{eq:bound_uN}) and the first statement of \eqref{eq:aNetaRN} are proved. For $N \in \N$ and $R \in \R$, we set
$$
h_N(R) := \sum_{k =N+1}^{+\infty} \frac{\cos(2 \pi k R)}{4\pi^2k^2} = \frac{1}{4\pi^2N}\eta_{R,N}.
$$
We have
\begin{align*}
|\psi_N(R)|=\left| g_N(R)-h_N(R) \right| &= \left|  -  \sum_{k=N+1}^{+\infty}   \frac{k_R^2\cos(2\pi k R)}{4\pi^2k^2(k_R^2+4\pi^2k^2)}\right|  \le M \sum_{k=N+1}^{+\infty}\frac{1}{16\pi^4 k^4} \le \frac{M}{48\pi^4 N^3}.
\end{align*}
Taking the second derivative of $h_N$ in the distribution sense and using Poisson summation formula, we obtain
\begin{align*}
h''_N(R) &= \frac{d^2}{dR^2} \left( \sum_{k =N+1}^{+\infty} \frac{e^{2 i \pi k R}+ e^{-2 i \pi k R}}{8\pi^2k^2}\right) = - \frac 12 \left( \sum_{k \in \Z \, | \, |k| > N} e^{2 i \pi k R} \right) \\
&= - \frac 12 \left( \sum_{k \in \Z} e^{2 i \pi k R} -  \sum_{k=-N}^N e^{2 i \pi k R} \right) = - \frac 12 \sum_{m \in \Z} \delta_m(R) + \frac 12 \frac{\sin\left( (2N+1)\pi R \right)}{\sin(\pi R)}.
\end{align*}
Therefore, $h_N$ is smooth on $\R \setminus \Z$. Since it is 1-periodic, it suffices to study it on the open interval $(0,1)$. Since $h_N\left(\frac 12+t \right) = h_N\left(\frac 12-t \right)$ for all $|t|<\frac 12$, we have $h_N'\left(\frac 12\right) =0$, so that for all $R \in (0,1)$, and using Taylor formula with integral remainder, we get
\begin{align*}
h_N(R) &= h_N\left(\frac 12 \right) + \int_{\frac 12}^R \left( R-t \right) h_N''(t) \; dt = h_N\left(\frac 12 \right) + \frac 12  \int_{\frac 12}^R \left( R-t \right)  \frac{\sin\left( (2N+1)\pi t \right)}{\sin(\pi t)}  \; dt \\
&= h_N\left(\frac 12 \right) + \frac{1}{2(2N+1)^2\pi^2} \left( (-1)^N - \frac{\sin\left( (2N+1)\pi R \right)}{\sin(\pi R)} \right) \\
& \quad -  \frac{1}{2(2N+1)^2\pi^2}  \int_{\frac 12}^R 
\left( 2\pi \frac{\cos(\pi t)}{\sin(\pi t)} + \frac{(R-t)\pi^2(1+\cos^2(\pi t))}{\sin^2(\pi t)}  \right) \frac{\sin\left( (2N+1)\pi t \right)}{\sin(\pi t)}  \; dt.
\end{align*}
Since
$$
\left| h_N\left(\frac 12\right) \right| = \left| \sum_{k=N+1}^{+\infty} \frac{(-1)^k}{4\pi^2k^2} \right| \le \frac{1}{4\pi^2(N+1)^2} \le \frac{1}{4\pi^2N^2},
$$
and since, for all $R \in (0,1/2)$, 
$$
\left| \frac{1}{2(2N+1)^2\pi^2} \left( (-1)^N - \frac{\sin\left( (2N+1)\pi R \right)}{\sin(\pi R)} \right) \right| \le \frac{1}{8\pi^2N^2} \left( 1+\frac{1}{\sin(\pi R)} \right) \le \frac{1}{4\pi^2N^2\sin(\pi R)},
$$
\begin{align*}
\left| \int_{\frac 12}^R
 2\pi \frac{\cos(\pi t)}{\sin(\pi t)}  \frac{\sin\left( (2N+1)\pi t \right)}{\sin(\pi t)}  \; dt \right| &\le 
 2\pi  \int_R^{\frac 12} \frac{\cos(\pi t)}{\sin^2(\pi t)} dt = 2 \left( \frac{1}{\sin(\pi R)}-1 \right),
\end{align*}
and, using the inequalities $2t < \sin(\pi t) <\pi t$ for all $0 < t < \frac 12$,
\begin{align*}
\left| \int_{\frac 12}^R
 \frac{(R-t)\pi^2(1+\cos^2(\pi t))}{\sin^2(\pi t)}  \frac{\sin\left( (2N+1)\pi t \right)}{\sin(\pi t)}  \; dt \right| &\le  2\pi^2  \int_R^{\frac 12}
 \frac{t-R}{\sin^3(\pi t)}    \; dt  \le  \pi^2  \int_R^{\frac 12}  \frac{2t}{\sin^3(\pi t)}    \; dt  \\ 
&\le  \frac{\pi^2}4  \int_R^{\frac 12} \frac{1}{t^2} dt \le \frac{\pi^2}{4R} \le \frac{\pi^3}{4 \sin(\pi R)},
\end{align*}
we finally get
\begin{align*}
|\eta_{N,R}| = \left| 4\pi^2N h_N(R) \right|  & \le  \frac{1}{N}+  \frac{1}{N\sin(\pi R)}+ \frac 1 N  \left( \frac{1}{\sin(\pi R)}-1 \right) + \frac{\pi^3}{8 \sin(\pi R) N} \\
& = \left( 2 + \frac{\pi^3}8 \right) \frac{1}{\sin(\pi R) N}.
\end{align*}
which concludes the proof.
\end{proof}

We are now ready to prove Theorem \ref{th:th1}.

\begin{myproof}[Proof of Theorem \ref{th:th1}]
Combining Lemmata \ref{lem:lem0}, \ref{lem:lem1}, \ref{lem:lem2} and \ref{lem:lem3}, we get that for any $R\in (0,1)$, 
\begin{align*}
E_{R,N}-E_R &= z_1 u_{R,N}(0)(\Pi_N^\perp u_R) (0) + z_2 u_{R,N}(R)(\Pi_N^\perp u_R) (R)  + o\left( \frac{1}{N^{3-\epsilon}}\right)  \quad & \mbox{(Lemma~\ref{lem:lem1})} \\
&=   z_1 u_{R,N}(0) \left( 2z_1u_R(0)f_N(R) + 2z_2u_R(R)g_N(R) \right) \\
&\quad + z_2 u_{R,N}(R) \left( 2z_2u_R(R)f_N(R) + 2z_1u_R(0)g_N(R) \right)  + o\left( \frac{1}{N^{3-\epsilon}}\right)  \quad & \mbox{(Lemma~\ref{lem:lem2})} \\
&=   \left( 2z_1^2 u_{R,N}(0) u_R(0) +  2z_2^2   u_{R,N}(R) u_R(R) \right) f_N(R)   \\
&\quad + 2z_1z_2 \left( u_{R,N}(0) u_R(R) + u_{R,N}(R) u_R(0) \right) g_N(R)  + o\left( \frac{1}{N^{3-\epsilon}}\right)  \quad &  \\
&=   \left( 2z_1^2 u_{R,N}(0) u_R(0) +  2z_2^2   u_{R,N}(R) u_R(R) \right) \frac{1}{4\pi^2 N}a_N   \\
&\quad + 2z_1z_2 \left( u_{R,N}(0) u_R(R) + u_{R,N}(R) u_R(0) \right) \frac{1}{4\pi^2 N}\eta_{R,N}  + o\left( \frac{1}{N^{3-\epsilon}}\right)  \quad & \mbox{(Lemma~\ref{lem:lem3})} \\
&=   \frac{\alpha_R}N a_N + \frac{\beta_{R,N}^{(1)}}N a_N + \frac{\gamma_R}{N^2} \eta_{R,N}   + o\left( \frac{1}{N^{3-\epsilon}}\right) ,
\end{align*}
where we have used the bounds \eqref{eq:uN(0)-u(0)} and \eqref{eq:aNetaRN} to obtain the last equality.
The proof of \eqref{eq:CVQN} easily follows.
\end{myproof}

\section*{Acknowledgments} The authors are grateful to Yvon Maday for useful discussions. 
This work was partially undertaken in the framework of CALSIMLAB, supported by the public
grant ANR-11-LABX- 0037-01 overseen by the French National Research Agency (ANR) as part
of the Investissements d'avenir program (reference: ANR-11-IDEX-0004-02).

\bibliography{biblio_tot}

\end{document}